\documentclass[conference]{IEEEtran}
\IEEEoverridecommandlockouts
	\makeatletter
	\def\ps@headings{%
	\def\@oddhead{\mbox{}\scriptsize\rightmark \hfil \thepage}%
	\def\@evenhead{\scriptsize\thepage \hfil \leftmark\mbox{}}%
	\def\@oddfoot{}%
	\def\@evenfoot{}}
	\makeatother
\usepackage{cite}
\usepackage[dvips]{graphicx}        
  \DeclareGraphicsExtensions{.eps}	
\usepackage{amssymb}
\usepackage[cmex10]{amsmath}
\usepackage{hyperref}
\hypersetup{letterpaper, pdfstartview=FitH, pdftitle={Optimal Scheduling for Fair Resource Allocation in Ad Hoc Networks with Elastic and Inelastic Traffic}, pdfauthor={Juan Jos\'e Jaramillo and R. Srikant}, bookmarksnumbered}
\newtheorem{theorem}{Theorem}
\newtheorem{lemma}{Lemma}
\newtheorem{corollary}{Corollary}
\newtheorem{fact}{Fact}
\begin{document}
\title{Optimal Scheduling for Fair Resource Allocation in Ad Hoc Networks with Elastic and Inelastic Traffic}
\author{
	\IEEEauthorblockN{Juan Jos\'e Jaramillo}
	\IEEEauthorblockA{Coordinated Science Laboratory and\\
	Department of Electrical and Computer Engineering\\
	University of Illinois at Urbana-Champaign\\
	Email: jjjarami@illinois.edu}
	\and
	\IEEEauthorblockN{R. Srikant}
	\IEEEauthorblockA{Coordinated Science Laboratory and\\
	Department of Electrical and Computer Engineering\\
	University of Illinois at Urbana-Champaign\\
	Email: rsrikant@illinois.edu}
	\thanks{Research supported by NSF Grants 07-21286, 05-19691, 03-25673, ARO MURI Subcontracts, AFOSR Grant FA-9550-08-1-0432 and DTRA Grant HDTRA1-08-1-0016.}
}
\maketitle

%
%
\begin{abstract}

This paper studies the problem of congestion control and scheduling in ad hoc wireless networks that have to support a mixture of best-effort and real-time traffic. Optimization and stochastic network theory have been successful in designing architectures for fair resource allocation to meet long-term throughput demands. However, to the best of our knowledge, strict packet delay deadlines were not considered in this framework previously. In this paper, we propose a model for incorporating the quality of service (QoS) requirements of packets with deadlines  in the optimization framework. The solution to the problem results in a joint congestion control and scheduling algorithm which fairly allocates resources to meet the fairness objectives of both elastic and inelastic flows, and per-packet delay requirements of inelastic flows.

\end{abstract}
%
%
\section{Introduction}

As wireless networks become more prevalent, they will be expected to support a wide variety of services, including best-effort and real-time traffic. Such networks will have to serve flows that require quality of service requirements, such as minimum bandwidth and maximum delay constraints, while at the same time keeping the network queues stable for data traffic and guaranteeing throughput optimality. For the case of wireless networks with best-effort traffic only, optimization-based algorithms  which naturally map into different layers of the protocol stack have been proposed in the last few years \cite{Eryilmaz05, Lin04, Neely05, Stolyar05, Eryilmaz06, Chen06}; see \cite{Lin06} for a survey. However, these models do not take into account strict per-packet delay bounds.

Scheduling packets with strict deadlines has been studied in \cite{Shakkottai02, Raghunathan08, Dua07, Liu06}, but all of these papers provide approximate solutions. The model that we study in this paper builds upon the recent work in \cite{Hou09a, Hou09b, Hou09c} on admission control and scheduling for inelastic flows in collocated wireless networks, i.e., networks where all links interfere with each other. Among the many contributions in these papers is a key modeling innovation whereby  the network is studied in frames, where a frame is a contiguous set of time-slots of fixed duration. Packets with deadlines are assumed to arrive at the beginning of a frame and have to be served by the end of the frame. In this paper, we explore this modeling paradigm further to study the design of resource allocation algorithms for ad hoc networks. The frame-based model allows us to incorporate delay deadlines in the optimization framework for very general network models, and somewhat surprisingly, allows us to design a common framework for handling both elastic and inelastic flows.

The main contributions of the paper are as follows:

\begin{enumerate}
	\item We present an optimization framework for resource allocation in a wireless network consisting of both best-effort flows and flows that generate traffic with per-packet delay constraints. The framework allows for very general interference, channel and arrival models.
	\item Using a dual decomposition approach, we derive an optimal scheduling and congestion control algorithm that fairly allocates resources and ensures that a required fraction of each inelastic flow's packets are delivered on time by appealing to connections between Lagrange multipliers, queues, and service deficits. The scheduling algorithm seamlessly integrates inelastic and elastic traffic into a unified max-weight scheduling framework, extending the well-known results in \cite{Tassiulas92}.
	\item The convergence of the above algorithm in an appropriate stochastic sense is proved and it is also shown that the network is stable.
\end{enumerate}

%
%
\section{Network Model}

The network is represented by a directed graph $\mathcal{G}=( \mathcal{N},\mathcal{L} )$, where $\mathcal{N}$ is the set of nodes and $\mathcal{L}$ is the set of directional links such that for all $n_1,n_2 \in \mathcal{N}$ if $(n_1,n_2) \in \mathcal{L}$ then node $n_1$ can transmit to node $n_2$. The links are numbered $1$ through $|\mathcal{L}|$, and by abusing notation, we sometimes use $l \in \cal L$ to mean $l\in\{1,2,\ldots, |\mathcal{L}|\}$.

Traffic is assumed to be a mixture of elastic and inelastic flows, where an inelastic flow is one that has maximum per-packet delay requirements. In contrast, elastic flows do not have such requirements.

Time is divided in slots, where a set of $T$ consecutive time slots makes a \emph{frame}. We assume that packet arrivals only occur at the beginning of a frame, and every inelastic packet has a deadline of $T$ time slots. If a packet misses its deadline it is discarded, and it is required that the loss probability at link $l \in \mathcal{L}$ due to deadline expiry must be no more than $p_l$. For elastic traffic we associate a utility function $U_l(x_l)$ which is a function of the mean elastic arrival rate per frame $x_l$. We assume that $U_l(.)$ is a concave function.

For a given frame, we denote by the vector $a_i=( a_{il} )_{l \in \mathcal{L}}$ the number of inelastic packet arrivals at every link, where $a_{il}$ is a random variable with mean $\lambda_l$ and variance $\sigma^2_{il}$. We further assume that arrivals are independent between different frames and that $Pr(a_{il}=0)>0$ and $Pr(a_{il}=1)>0$. The last two assumptions are used to guarantee that the Markov chain we define later is both irreducible and aperiodic, although these can be replaced by other similar assumptions. Similarly, we define $a_e=( a_{el} )_{l \in \mathcal{L}}$ to be the number of elastic packet arrivals at every link in a given frame.

The channel state is assumed to be constant in a given frame, independent between different frames, and independent of arrivals. The vector $c=( c_l )_{l \in \mathcal{L}}$ denotes the number of packets link $l$ can successfully transmit on a time slot in a given frame.

Depending on the wireless technology used, we can have some channel feedback before or after a transmission occurs. If channel estimation is performed before transmitting, we can determine the optimal rate at which we can successfully transmit.  Alternatively, feedback from the receiver after the transmission can be used to detect if a transmission is successful or not. In this paper we try to capture both scenarios in the following cases:

\begin{enumerate}
	\item Known channel state: It is assumed that $c_l$ is a \emph{non-negative} random variable with mean $\bar{c}_l$ and variance $\sigma_{cl}^2$, and we get to know the channel state at the \emph{beginning} of the frame.
	\item Unknown channel state: It is assumed that $c_l$ is a \emph{Bernoulli} random variable with mean $\bar{c}_l$ and we only get to know the channel state at the \emph{end} of the frame.
\end{enumerate}

In the known channel state case where we do channel estimation to determine the optimal transmission rate, we can potentially send more than one packet in a time slot at higher rates. This is captured by the fact that we make no assumptions on the values $c_l$ can take since it will be determined by the particular wireless technology used. In the case of unknown channel state we assume that we only get the binary feedback of acknowledgments, which is reflected in the Bernoulli assumption on $c_l$. In this case, and without any loss of generality, we assume only one packet can be transmitted per time slot per link.

In the rest of the paper we will consider the known channel case first and then in Section \ref{unknown_channel_case} we will highlight the differences in the analysis for the unknown channel case.

%
%
\section{Problem Formulation}
\label{formulation_kc}

We first formulate the problem as a static optimization problem. Using decomposition theory, we will then obtain a dynamic solution to this problem and prove its stability using stochastic Lyapunov techniques.

A feasible schedule $s=( s_{il,t}, s_{el,t} )$ is such that $s_{il,t}$, $s_{el,t}$ respectively denote the number of inelastic and elastic packets that can be scheduled for transmission at link $l \in \mathcal{L}$ and time slot $t \in \{ 1, 2, \ldots, T \}$; thus, $s_{il,t} + s_{el,t} > 0$ means that link $l$ is scheduled to transmit in time slot $t$ of the frame. Furthermore, for any $t$, if $s_{il_1,t} + s_{el_1,t} > 0$ and $s_{il_2,t} + s_{el_2,t} > 0$ then links $l_1$ and $l_2$ can be scheduled to simultaneously transmit without interfering with each other. Assuming the inelastic arrivals and the channel state are given by $a_i$ and $c$ respectively, we have the following constraints:
\begin{equation}
\label{constraint_inelastic_kc}
 \sum_{t=1}^{T} s_{il,t} \leq a_{il} \mbox{ for all } l \in \mathcal{L} \mbox{ and }
\end{equation}
\begin{equation}
\label{constraint_slot_kc}
 s_{il,t} + s_{el,t} \leq c_l \mbox{ for all } l \in \mathcal{L} \mbox{ and } t \in \{ 1, 2, \ldots, T \}.
\end{equation}
We denote by $\mathcal{S}(a_i, c)$ the set of all feasible schedules when the arrival state is $a_i$ and the channel state is $c$; thus, $\mathcal{S}(a_i, c)$ captures any interference constraints we have on our network and satisfies (\ref{constraint_inelastic_kc}) and (\ref{constraint_slot_kc}).

At the beginning of any frame we must choose a feasible schedule to serve all links and decide how many elastic packets are allowed to be injected in the network. Therefore, our goal is to find a function $Pr( s | a_i, c)$ which is the probability of using schedule $s \in \mathcal{S}(a_i, c)$ when the inelastic arrivals are given by $a_i$ and the channel state is $c$, subject to the constraint that the loss probability at link $l \in \mathcal{L}$ due to deadline expiry cannot exceed $p_l$. For elastic traffic, we want to select the vector $a_e$ such that we maximize the network utility while keeping the queues stable.

To properly formulate the problem, let us first define $\mu_{i}(a_i,c)$ to be the expected number of inelastic packets served if the number of packet arrivals is given by $a_i$ and the channel state is $c$. Similarly, $\mu_{e}(a_i,c)$ denotes the expected number of elastic packets that can be served. Therefore, we have the following constraints:
\begin{equation*}
 \mu_{il}(a_i,c) \leq \sum\limits_{s \in \mathcal{S}(a_i, c)} \sum_{t=1}^{T} s_{il,t} Pr(s | a_i, c)
\end{equation*}
\begin{equation*}
 \mu_{el}(a_i,c) \leq \sum\limits_{s \in \mathcal{S}(a_i, c)} \sum_{t=1}^{T} s_{el,t} Pr(s | a_i, c).
\end{equation*}

The expected service for mixed traffic at link $l$ is then given by
\begin{equation*}
	\mu_{il} \stackrel{def}{=} \sum\limits_{a_i} \sum\limits_{c} \mu_{il}(a_i,c) Pr(c) Pr(a_i)
\end{equation*}
\begin{equation*}
	\mu_{el} \stackrel{def}{=} \sum\limits_{a_i} \sum\limits_{c} \mu_{el}(a_i,c) Pr(c) Pr(a_i)
\end{equation*}
and due to QoS requirements and capacity constraints, we require that
\begin{equation*}
	\mu_{il} \geq \lambda_l(1-p_l) \mbox{ and } x_l \leq \mu_{el}.
\end{equation*}

We will focus on maximizing the following objective for some given vector $w \in \mathbb{R}_+^{|\mathcal{L}|}$:
\begin{equation}
\label{offline_opt_kc}
	\mathop{ \max\limits_{ \mu_{i}(a_i,c), \mu_{e}(a_i,c), } }\limits_{ \mu_{i}, \mu_{e}, x, Pr(s | a_i, c) } \sum\limits_{l \in \mathcal{L}} U_l(x_l) + w_l \mu_{il}
\end{equation}
subject to
\begin{equation*}
 \mu_{il}(a_i,c) \leq \sum\limits_{s} \sum_{t=1}^{T} s_{il,t} Pr(s | a_i, c) \mbox{ for all } l \in \mathcal{L}, a_i, c
\end{equation*}
\begin{equation*}
 \mu_{el}(a_i,c) \leq \sum\limits_{s} \sum_{t=1}^{T} s_{el,t} Pr(s | a_i, c) \mbox{ for all } l \in \mathcal{L}, a_i, c
\end{equation*}
\begin{equation*}
	\mu_{il} = \sum\limits_{a_i} \sum\limits_{c} \mu_{il}(a_i,c) Pr(c) Pr(a_i) \mbox{ for all } l \in \mathcal{L}
\end{equation*}
\begin{equation*}
	\mu_{el} = \sum\limits_{a_i} \sum\limits_{c} \mu_{el}(a_i,c) Pr(c) Pr(a_i) \mbox{ for all } l \in \mathcal{L}
\end{equation*}
\begin{equation*}
	\mu_{il} \geq \lambda_l(1-p_l) \mbox{ for all } l \in \mathcal{L}
\end{equation*}
\begin{equation*}
	0 \leq x_l \leq \mu_{el} \mbox{ for all } l \in \mathcal{L}
\end{equation*}
\begin{equation*}
	Pr(s | a_i, c) \geq 0 \mbox{ for all } s \in \mathcal{S}(a_i, c), a_i, c
\end{equation*}
\begin{equation*}
	\sum_{s} Pr(s | a_i, c) \leq 1 \mbox{ for all } a_i, c.
\end{equation*}

The vector $w$ can be used to allocate additional bandwidth fairly to inelastic flows beyond what is required to meet their QoS needs. Other uses for $w$ will be explored in the simulations section. We will assume that the arrivals and loss probability requirements are feasible and thus the optimization problem has a solution $(x^*, \mu_i^*)$.

%
%
\section{Solution Using Dual Decomposition}
\label{decomposition_kc}

Using the definition of the dual function\cite{Luenberger03}, we have that $D(\delta_i, \delta_e)=$
\begin{equation*}
	\mathop{ \max\limits_{ \mu_{i}(a_i,c), \mu_{e}(a_i,c), } }\limits_{ \mu_{i}, \mu_{e}, x, Pr(s | a_i, c) } \sum\limits_{l \in \mathcal{L}} \left\{ \begin{array}{l} U_l(x_l) + w_l \mu_{il} - \delta_{el} [x_l - \mu_{el}] \\ - \delta_{il} [\lambda_l(1-p_l) - \mu_{il}] \\ \end{array} \right\}
\end{equation*}
subject to
\begin{equation*}
 \mu_{il}(a_i,c) \leq \sum\limits_{s} \sum_{t=1}^{T} s_{il,t} Pr(s | a_i, c) \mbox{ for all } l \in \mathcal{L}, a_i, c
\end{equation*}
\begin{equation*}
 \mu_{el}(a_i,c) \leq \sum\limits_{s} \sum_{t=1}^{T} s_{el,t} Pr(s | a_i, c) \mbox{ for all } l \in \mathcal{L}, a_i, c
\end{equation*}
\begin{equation*}
	\mu_{il} = \sum\limits_{a_i} \sum\limits_{c} \mu_{il}(a_i,c) Pr(c) Pr(a_i) \mbox{ for all } l \in \mathcal{L}
\end{equation*}
\begin{equation*}
	\mu_{el} = \sum\limits_{a_i} \sum\limits_{c} \mu_{el}(a_i,c) Pr(c) Pr(a_i) \mbox{ for all } l \in \mathcal{L}
\end{equation*}
\begin{equation*}
	x_l \geq 0 \mbox{ for all } l \in \mathcal{L}
\end{equation*}
\begin{equation*}
	Pr(s | a_i, c) \geq 0 \mbox{ for all } s \in \mathcal{S}(a_i, c), a_i, c
\end{equation*}
\begin{equation*}
	\sum_{s} Pr(s | a_i, c) \leq 1 \mbox{ for all } a_i, c.
\end{equation*}

Slater's condition \cite{Boyd04} states that, since the objective is concave and the constraints are affine functions, the duality gap is zero and therefore $D(\delta_i^*, \delta_e^*)=\sum\limits_{l \in \mathcal{L}} U_l(x_l^*) + w_l \mu_{il}^*$, where
\begin{equation*}
	(\delta_i^*, \delta_e^*) \in \mathop{\arg\min}_{\delta_{il} \geq 0, \delta_{el} \geq 0} D(\delta_i, \delta_e).
\end{equation*}

We are interested in finding $(x^*, \mu_i^*)$ but not the value $D(\delta_i^*, \delta_e^*)$, so if we rewrite the objective in the dual function as
\begin{equation*}
	\mathop{ \max\limits_{ \mu_{i}(a_i,c), \mu_{e}(a_i,c), } }\limits_{ \mu_{i}, \mu_{e}, x, Pr(s | a_i, c) } \left\{ \begin{array}{l} \sum\limits_{l \in \mathcal{L}} U_l(x_l) - \delta_{el} x_l \\ + \sum\limits_{l \in \mathcal{L}} (w_l+\delta_{il}) \mu_{il} + \delta_{el} \mu_{el} \\ - \sum\limits_{l \in \mathcal{L}} \delta_{il} \lambda_l(1-p_l) \\ \end{array} \right\}
\end{equation*}
we notice that the problem can be decomposed into the following subproblems:
\begin{equation*}
	\max\limits_{ x_l \geq 0 } U_l(x_l) - \delta_{el} x_l
\end{equation*}
and
\begin{equation}
\label{first_decomposition_kc}
	\mathop{ \max\limits_{ \mu_{i}(a_i,c), \mu_{e}(a_i,c), } }\limits_{ \mu_{i}, \mu_{e}, Pr(s | a_i, c) } \sum\limits_{l \in \mathcal{L}} (w_l+\delta_{il}) \mu_{il} + \delta_{el} \mu_{el}
\end{equation}
subject to
\begin{equation*}
 \mu_{il}(a_i,c) \leq \sum\limits_{s} \sum_{t=1}^{T} s_{il,t} Pr(s | a_i, c) \mbox{ for all } l \in \mathcal{L}, a_i, c
\end{equation*}
\begin{equation*}
 \mu_{el}(a_i,c) \leq \sum\limits_{s} \sum_{t=1}^{T} s_{el,t} Pr(s | a_i, c) \mbox{ for all } l \in \mathcal{L}, a_i, c
\end{equation*}
\begin{equation*}
	\mu_{il} = \sum\limits_{a_i} \sum\limits_{c} \mu_{il}(a_i,c) Pr(c) Pr(a_i) \mbox{ for all } l \in \mathcal{L}
\end{equation*}
\begin{equation*}
	\mu_{el} = \sum\limits_{a_i} \sum\limits_{c} \mu_{el}(a_i,c) Pr(c) Pr(a_i) \mbox{ for all } l \in \mathcal{L}
\end{equation*}
\begin{equation*}
	Pr(s | a_i, c) \geq 0 \mbox{ for all } s \in \mathcal{S}(a_i, c), a_i, c
\end{equation*}
\begin{equation*}
	\sum_{s} Pr(s | a_i, c) \leq 1 \mbox{ for all } a_i, c.
\end{equation*}
Furthermore, since we are interested in solving the problem for non-negative values of $\delta_{il}$ and $\delta_{el}$, it must be the case that $\mu_{i}^*$ and $\mu_{e}^*$ are as large as the constraints allow, and since the upper bounds for $\mu_{il}^*(a_i,c)$ and $\mu_{el}^*(a_i,c)$ are expressed as a convex combination, and the objective function in (\ref{first_decomposition_kc}) is linear, the problem can be decomposed into the following subproblems for fixed $a_i$ and $c$:
\begin{equation*}
	\max\limits_{s \in \mathcal{S}(a_i,c)} \sum\limits_{l \in \mathcal{L}} \left\{ (w_l+\delta_{il}) \sum_{t=1}^{T} s_{il,t} + \delta_{el} \sum_{t=1}^{T} s_{el,t} \right\}.
\end{equation*}

This suggests the following iterative algorithm to find the solution to our optimization problem, where $k$ is the step index and $X_{max} > \max_{l \in \mathcal{L}} x_l^*$ is a fixed parameter:
\begin{equation*}
	\tilde{x}_l^*(k) \in \mathop{\arg\max}\limits_{ 0 \leq x_l \leq X_{max} } U_l(x_l) - \delta_{el}(k) x_l
\end{equation*}
\begin{align*}
	& \tilde{s}^*(a_i,c,k) \in \\
	& \mathop{\arg\max}\limits_{s \in \mathcal{S}(a_i,c)} \sum\limits_{l \in \mathcal{L}} \left\{ [w_l+\delta_{il}(k)] \sum_{t=1}^{T} s_{il,t} + \delta_{el}(k) \sum_{t=1}^{T} s_{el,t} \right\}
\end{align*}
\begin{equation*}
	\tilde{\mu}_{il}^*(k) = \sum\limits_{a_i} \sum\limits_{c} \sum_{t=1}^{T} \tilde{s}_{il,t}^*(a_i, c, k) Pr(c) Pr(a_i)
\end{equation*}
\begin{equation*}
	\tilde{\mu}_{el}^*(k) = \sum\limits_{a_i} \sum\limits_{c} \sum_{t=1}^{T} \tilde{s}_{el,t}^*(a_i, c, k) Pr(c) Pr(a_i).
\end{equation*}
We update the Lagrange multipliers $\delta_i(k)$, $\delta_e(k)$ at every step according to the following equations:
\begin{equation*}
	\delta_{il}(k+1) = \{ \delta_{il}(k) + \epsilon [\lambda_l(1-p_l) - \tilde{\mu}_{il}^*(k)] \}^+
\end{equation*}
and
\begin{equation*}
	\delta_{el}(k+1) = \{ \delta_{el}(k) + \epsilon [\tilde{x}_l^*(k) - \tilde{\mu}_{el}^*(k)] \}^+
\end{equation*}
where $\epsilon>0$ is a fixed step-size parameter, and for any $\alpha \in \mathbb{R}$, $\alpha^+ \stackrel{def}{=}\max\{ \alpha, 0 \}$.

Making the change of variables $\epsilon \hat{d}(k) = \delta_i(k)$ and $\epsilon \hat{q}(k) = \delta_e(k)$, we have that our iterative algorithm can be rewritten as
\begin{equation*}
	\tilde{x}_l^*(k) \in \mathop{\arg\max}\limits_{ 0 \leq x_l \leq X_{max} } \frac{1}{\epsilon} U_l(x_l) - \hat{q}_l(k) x_l
\end{equation*}
\begin{align*}
	& \tilde{s}^*(a_i,c,k) \in \\
	& \mathop{\arg\max}\limits_{s \in \mathcal{S}(a_i,c)} \sum\limits_{l \in \mathcal{L}} \left\{ [\frac{1}{\epsilon} w_l+\hat{d}_l(k)] \sum_{t=1}^{T} s_{il,t}  + \hat{q}_l(k) \sum_{t=1}^{T} s_{el,t} \right\}
\end{align*}
\begin{equation*}
	\tilde{\mu}_{il}^*(k) = \sum\limits_{a_i} \sum\limits_{c} \sum_{t=1}^{T} \tilde{s}_{il,t}^*(a_i, c, k) Pr(c) Pr(a_i)
\end{equation*}
\begin{equation*}
	\tilde{\mu}_{el}^*(k) = \sum\limits_{a_i} \sum\limits_{c} \sum_{t=1}^{T} \tilde{s}_{el,t}^*(a_i, c, k) Pr(c) Pr(a_i).
\end{equation*}
with update equations
\begin{equation*}
	\hat{d}_l(k+1) = [ \hat{d}_l(k) + \lambda_l(1-p_l) - \tilde{\mu}_{il}^*(k)]^+
\end{equation*}
\begin{equation*}
	\hat{q}_l(k+1) = [ \hat{q}_l(k) + \tilde{x}_l^*(k) - \tilde{\mu}_{el}^*(k) ]^+.
\end{equation*}

It should be noted that due to the change of variables $\hat{d}_l(k)$ can be interpreted as a queue that has $\lambda_l(1-p_l)$ arrivals and $\tilde{\mu}_{il}^*(k)$ departures at step $k$; $\hat{q}_l(k)$ can have a similar queue interpretation. The dual decomposition approach only provides an intuition behind the solution, but the real network has stochastic and dynamic arrivals and channel state conditions. In the next section, we present the complete solution which takes into account these dynamics and we also establish its convergence properties.

%
%
\section{Dynamic Algorithm and Its Convergence Analysis}
\subsection{Scheduler and Congestion Controller}
\label{online_kc}

To implement the algorithm online, we propose the following congestion control algorithm in frame $k$, where the queue length at link $l$ is given by $q_l(k)$:
\begin{equation}
\label{online_opt_cc_kc}
	\tilde{x}_l^*(k) \in \mathop{\arg\max}\limits_{ 0 \leq x_l \leq X_{max} } \frac{1}{\epsilon} U_l(x_l) - q_l(k) x_l.
\end{equation}

We need to convert this elastic arrival rate, which in general is a non-negative real number, into a non-negative integer indicating the number of elastic packets allowed to enter the network in a given frame. This conversion can be made in many different ways: we assume the elastic arrivals at link $l$, $\tilde{a}_{el}(k)$, are a random variable with mean $\tilde{x}_l^*(k)$ and variance upper-bounded by $\sigma_{e}^2$, and are such that $Pr(\tilde{a}_{el}(k)=0)>0$ and $Pr(\tilde{a}_{el}(k)=1)>0$ for all $l \in \mathcal{L}$ and all $k$. The last two assumptions are used to guarantee the Markov chain we define below is both irreducible and aperiodic, although these can be replaced by other similar assumptions.

Letting the number of inelastic arrivals be denoted by $a_i(k)$ and the channel state by $c(k)$, we propose the following scheduling algorithm:
\begin{align}
	& \tilde{s}^*(a_i(k),c(k),d(k),q(k)) \in \label{online_opt_sch_kc} \\
	& \mathop{\arg\max}\limits_{s \in \mathcal{S}(a_i(k),c(k))} \sum\limits_{l \in \mathcal{L}} \left\{ [\frac{1}{\epsilon} w_l+d_l(k)] \sum_{t=1}^{T} s_{il,t} + q_l(k) \sum_{t=1}^{T} s_{el,t} \right\}. \notag
\end{align}

The vectors $d(k)$ and $q(k)$ are updated from frame to frame as follows:
\begin{equation*}
	d_l(k+1) = [ d_l(k) + \tilde{a}_{il}(k) - I_{il}^*(a_i(k),c(k),d(k),q(k))]^+
\end{equation*}
\begin{equation*}
	q_l(k+1) = [ q_l(k) + \tilde{a}_{el}(k) - I_{el}^*(a_i(k),c(k),d(k),q(k)) ]^+,
\end{equation*}
where
\begin{equation*}
	I_{il}^*(a_i(k),c(k),d(k),q(k)) = \sum_{t=1}^{T} \tilde{s}_{il,t}^*(a_i(k),c(k),d(k),q(k))
\end{equation*}
\begin{equation*}
	I_{el}^*(a_i(k),c(k),d(k),q(k)) = \sum_{t=1}^{T} \tilde{s}_{el,t}^*(a_i(k),c(k),d(k),q(k))
\end{equation*}
and $\tilde{a}_{il}(k)$ is a binomial random variable with parameters $a_{il}(k)$ and $1-p_l$. The quantity $\tilde{a}_{il}(k)$ can be generated by the network as follows: upon each inelastic packet arrival, toss a coin with probability of \emph{heads} equal to $1-p_l$, and if the outcome is \emph{heads}, add a one to the deficit counter.

In our notation we make explicit the fact that for fixed $\epsilon$ and $w$, the optimal scheduler (\ref{online_opt_sch_kc}) is a function of $a_i(k)$, $c(k)$, $d(k)$, and $q(k)$. We interpret $d_l(k)$ as a virtual queue that counts the deficit in service for link $l$ to achieve a loss probability due to deadline expiry less than or equal to $p_l$. This deficit queue was first used in the inelastic traffic context in \cite{Hou09a} for the case of collocated networks; the connection to the dual decomposition approach now provides a Lagrange multiplier interpretation to it and allows the extension to general ad hoc networks. Note that $q_l(k)$ is just the queue size for elastic packets at link $l$.

\subsection{Convergence Results}
\label{convergence_kc}

For readability, we present the main results in this section, but the proofs are deferred to the appendixes. We start by noting that $(d(k),q(k))$ defines an irreducible and aperiodic Markov chain. To prove that our dynamic algorithm achieves the optimal solution to the static problem (\ref{offline_opt_kc}) in some average sense and fulfills all links' requirements, we will first bound the expected drift of $(d(k),q(k))$ for a suitable Lyapunov function.

\begin{lemma}
\label{expected_drift_kc}
Consider the Lyapunov function $V(d,q)=\frac{1}{2}\sum_{l \in \mathcal{L}}d_l^2+q_l^2$. If $\mu^*_{il} > \lambda_l(1-p_l)$ and $\mu^*_{el} > x_l^*$ for all $l \in \mathcal{L}$, then
\begin{align*}
	E & \left[ V(d(k+1),q(k+1)) | d(k)=d, q(k)=q \right] - V(d,q) \\
 	\leq & B_1 - B_2 \sum_{l \in \mathcal{L}} d_l - B_3 \sum_{l \in \mathcal{L}} q_l - \frac{1}{\epsilon} \sum_{l \in \mathcal{L}} [ U_l(x_l^*) - U_l(\tilde{x}_l^*(k)) ] \\
 	& - \frac{1}{\epsilon} \sum_{l \in \mathcal{L}} w_l \mu^*_{il} - w_l E \left[ I_{il}^*(a_i(k),c(k),d,q) \right]
\end{align*}
for some positive constants $B_1$, $B_2$, $B_3$, any $\epsilon>0$, where $(x^*, \mu_i^*)$ is the solution to (\ref{offline_opt_kc}), $\tilde{x}^*(k)$ is the solution to (\ref{online_opt_cc_kc}), and $I_i^*(a_i(k),c(k),d,q)$ is obtained from the solution to (\ref{online_opt_sch_kc}).
$\hfill \diamond$
\end{lemma}

It is important to note that since the last two terms in the right-hand side of the inequality can be upper-bounded, Lemma \ref{expected_drift_kc} implies that $(d(k),q(k))$ is positive recurrent since the expected drift is negative but for a finite set of values of $(d(k),q(k))$. As a direct consequence of this fact, we note that the total service deficit and queue length have a $O(1 / \epsilon)$ bound.

\begin{corollary}
\label{def_queue_bound_kc}
If $\mu^*_{il} > \lambda_l(1-p_l)$ and $\mu^*_{el} > x_l^*$ for all $l \in \mathcal{L}$, then the total expected service deficit and network queue length is upper-bounded by
\begin{equation*}
	\limsup_{k \rightarrow \infty} E\left[ \sum_{l \in \mathcal{L}} d_l(k)+q_l(k) \right] \leq B_4 + \frac{1}{\epsilon} B_5
\end{equation*}
for all $l \in \mathcal{L}$ and
\begin{equation*}
	B_4 = \frac{B_1}{ \min \{ B_2, B_3 \} }
\end{equation*}
and
\begin{equation*}
	B_5 \leq \frac{ \sum_{l \in \mathcal{L}} \max_{ 0 \leq x_l \leq X_{max} } 2 |U_l(x_l)| + w_l \lambda_l }{ \min \{ B_2, B_3 \} }.
\end{equation*}
$\hfill \diamond$
\end{corollary}

This also implies that the scheduling and congestion control algorithm fulfills all links' inelastic requirements.
\begin{corollary}
If $\mu^*_{il} > \lambda_l(1-p_l)$ and $\mu^*_{el} > x_l^*$ for all $l \in \mathcal{L}$, then the online algorithm fulfills all the inelastic constraints. That is:
\begin{equation*}
	\liminf_{K \rightarrow \infty} E \left[ \frac{1}{K} \sum_{k=1}^{K} I_{il}^*(a_i(k),c(k),d(k),q(k)) \right] \geq \lambda_l(1-p_l)
\end{equation*}
for all $l \in \mathcal{L}$.
$\hfill \diamond$
\end{corollary}

The above corollary simply states that the arrival rate into the deficit counter is less than or equal to the departure rate. This result is an obvious consequence of the stability of the deficit counters and so a formal proof is not provided here.

Now we are ready to prove that our online algorithm is within $O(\epsilon)$ of the optimal value.
\begin{theorem}
\label{optimality_online_kc}
For any $\epsilon >0$, if $\mu^*_{il} > \lambda_l(1-p_l)$ and $\mu^*_{el} > x_l^*$ for all $l \in \mathcal{L}$, then
\begin{align*}
	\limsup_{K \rightarrow \infty} & E \left[ \sum_{l \in \mathcal{L}} U_l(x_l^*) + w_l \mu^*_{il} - \sum_{l \in \mathcal{L}} \frac{1}{K} \sum_{k=1}^K U_l(\tilde{x}_l^*(k)) \right. \\
	& \left. - \sum_{l \in \mathcal{L}} \frac{1}{K} \sum_{k=1}^K w_l I_{il}^*(a_i(k),c(k),d(k),q(k)) \right] \leq B \epsilon
\end{align*}
for some $B>0$, where $(x^*, \mu_i^*)$ is the solution to (\ref{offline_opt_kc}), $\tilde{x}^*(k)$ is the solution to (\ref{online_opt_cc_kc}), and $I_i^*(a_i(k),c(k),d(k),q(k))$ is obtained from the solution to (\ref{online_opt_sch_kc}).
$\hfill \diamond$
\end{theorem}

In conclusion, there is a trade-off in choosing the parameter $\epsilon$: smaller values will achieve a solution closer to the optimal, but at the same time the deficit in service at the links and the aggregate queue length increase. The statement and the proof of Theorem \ref{optimality_online_kc} follows the techniques in \cite{Neely05}. The result can also be derived, in a slightly different form, using the techniques in \cite{Stolyar05}. A closely related result can be obtained using the methods in \cite{Eryilmaz05}.

%
%
\section{Unknown Channel State}
\label{unknown_channel_case}

The analysis for the unknown channel case is similar to the one we presented for the known channel case, so in this section we will only highlight the differences.

A feasible schedule $s=( s_{il,t}, s_{el,t} )$ is such that $s_{il,t}$, $s_{el,t}$ respectively denote the number of inelastic and elastic packets that can be scheduled for transmission at link $l \in \mathcal{L}$ and time $t \in \{ 1, 2, \ldots, T \}$ without violating any interference constraints. Assuming the inelastic arrivals are given by $a_i$, and since we can only schedule at most one packet per link at every time slot, we have the following constraints:
\begin{equation}
\label{constraint_inelastic_uc}
 \sum_{t=1}^{T} s_{il,t} \leq a_{il} \mbox{ for all } l \in \mathcal{L} \mbox{ and }
\end{equation}
\begin{equation}
\label{constraint_slot_uc}
 s_{il,t} + s_{el,t} \leq 1 \mbox{ for all } l \in \mathcal{L} \mbox{ and } t \in \{ 1, 2, \ldots, T \}.
\end{equation}
We denote by $\mathcal{S}(a_i)$ the set of all feasible schedules for fixed arrivals, capturing any interference constraints we have on our network, and satisfying (\ref{constraint_inelastic_uc}) and (\ref{constraint_slot_uc}).

Our goal now is to find a function $Pr(s | a_i)$ which is the probability of using schedule $s \in \mathcal{S}(a_i)$ when the inelastic arrivals are given by $a_i$, subject to the constraint that the loss probability at link $l \in \mathcal{L}$ due to deadline expiry cannot exceed $p_l$. For elastic traffic, we still want to select the vector $a_e$ such that we maximize the total utility while keeping the queues stable.

For a given distribution $Pr(s | a_i)$ we have that $\mu_{il}(a_i)$ is the expected number of attempted inelastic transmissions if arrivals are given by $a_i$. Similarly, $\mu_{el}(a_i)$ denotes the expected number of times link $l$ is scheduled to serve elastic packets in a given frame. As before, we have the following constraints:
\begin{equation*}
 \mu_{il}(a_i) \leq \sum\limits_{s} \sum_{t=1}^{T} s_{il,t} Pr(s | a_i)
\end{equation*}
\begin{equation*}
 \mu_{el}(a_i) \leq \sum\limits_{s} \sum_{t=1}^{T} s_{el,t} Pr(s | a_i)
\end{equation*}

%

When the (unknown) channel state is $c$, we have that $c_l \mu_{il}(a_i)$ is the expected number of successful inelastic transmissions per frame at link $l$ for fixed arrivals, while $c_l \mu_{el}(a_i)$ is the expected service to link $l$ for inelastic arrivals. Thus, the expected service for mixed traffic at link $l$ is given by
\begin{equation*}
	\mu_{il} \stackrel{def}{=} \sum\limits_{a_i} \sum\limits_{c} c_l \mu_{il}(a_i) Pr(c) Pr(a_i)
\end{equation*}
\begin{equation*}
	\mu_{el} \stackrel{def}{=} \sum\limits_{a_i} \sum\limits_{c} c_l \mu_{el}(a_i) Pr(c) Pr(a_i).
\end{equation*}

Simplifying both expressions we get
\begin{equation*}
	\mu_{il} = \sum\limits_{a_i} \bar{c}_l \mu_{il}(a_i) Pr(a_i)
\end{equation*}
\begin{equation*}
	\mu_{el} = \sum\limits_{a_i} \bar{c}_l \mu_{el}(a_i) Pr(a_i).
\end{equation*}

Due to service requirements and capacity constraints we need that
\begin{equation*}
	\mu_{il} \geq \lambda_l(1-p_l) \mbox{ and } x_l \leq \mu_{el}.
\end{equation*}
%

With the definitions and constraints stated above we can formulate the optimization problem in a similar way as in (\ref{offline_opt_kc}).
%
%

The only difference with the known channel state case is the scheduling algorithm. Assuming inelastic arrivals are given by $a_i(k)$ the scheduling algorithm is given by
\begin{align}
	& \tilde{s}^*(a_i(k),d(k),q(k)) \in \notag \\
	& \mathop{\arg\max}\limits_{s \in \mathcal{S}(a_i(k))} \sum\limits_{l \in \mathcal{L}} \left\{ [\frac{1}{\epsilon} w_l+d_l(k)] \bar{c}_l \sum_{t=1}^{T} s_{il,t} + q_l(k) \bar{c}_l \sum_{t=1}^{T} s_{el,t} \right\}. \notag
\end{align}

The main difference in the scheduling algorithm compared to the known channel state case is that the network now uses the expected channel state in making scheduling decisions. Thus, the network needs to know or estimate $\bar{c}_l$ as in \cite{Hou09a}.
%

Similar results can be proved for this algorithm using the techniques developed in Section \ref{convergence_kc}, whereby one can show that the algorithm meets all the inelastic QoS constraints, the total expected service deficits and the queue lengths have a $O(1 / \epsilon)$ bound, and the mean value of the objective is within $O(\epsilon)$ of the optimal value.

%
%


%
%
\section{Simulations}

The purpose of this simulation study is to understand how the parameter $\epsilon$ and the link weights $w_l$ impact the performance of the algorithm,  and how a greedy heuristic can be used to implement the optimal scheduler. We simulate a 10-link network with an interference graph given by Fig. \ref{interference_graph}, where each node represents a link and each edge means that the two adjacent links cannot be scheduled simultaneously. For example, if link 1 is scheduled, then links 2, 4, and 7 cannot be activated. The required loss probability due to deadline expiry of inelastic packets is set to 0.1, the link arrivals are assumed to have a Bernoulli distribution with mean 0.6 packets/frame, and there are 3 time slots per frame. The channel for every link is assumed to have a Bernoulli distribution with mean 0.96, and we get to know the channel state at the beginning of the frame. We set $U_l(x_l)=\log(x_l)$ for all links. The simulation time was $10^6$ frames.

\begin{figure}[t]
	\centering
	\includegraphics[angle=0, width=2.0in]{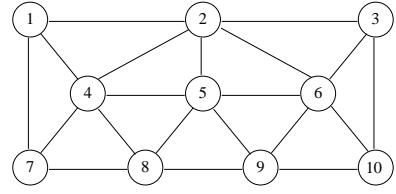}
	\caption{Interference graph used in the simulations}
	\label{interference_graph}
\end{figure}

As can be noted from (\ref{online_opt_sch_kc}), the max-weight scheduler requires that we do an exhaustive search to find the optimal schedule at every frame. For large networks this can become a burden due to the large search space; thus we explore a greedy heuristic and check how close it is to the optimal solution: at any given time slot, the greedy scheduler orders all links according to their weights. The greedy scheduler adds one of the links with the largest weight to the schedule, then removes all links that interfere with this link from the graph, then schedules a link with the largest weight among the remaining links, and so on. This procedure continues until no more links can be scheduled.


In Figs. \ref{def_queue_ten_w0}, \ref{def_queue_ten_w3}, and \ref{def_queue_ten_w6}, we plot the expected values of the deficit counters and queues per link for various values of $w_l$, and compare their evolution for both the scheduler with optimal decisions and the greedy scheduler.

\begin{figure}[t]
	\centering
	\includegraphics[angle=270, width=3.3in]{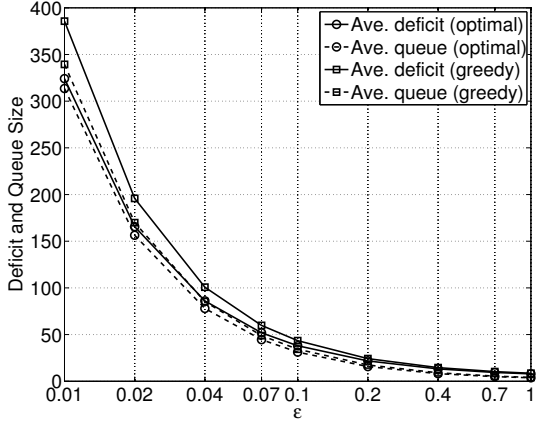}
	\caption{Deficit size and queue length when $w_l=0$}
	\label{def_queue_ten_w0}
\end{figure}
\begin{figure}[t]
	\centering
	\includegraphics[angle=270, width=3.3in]{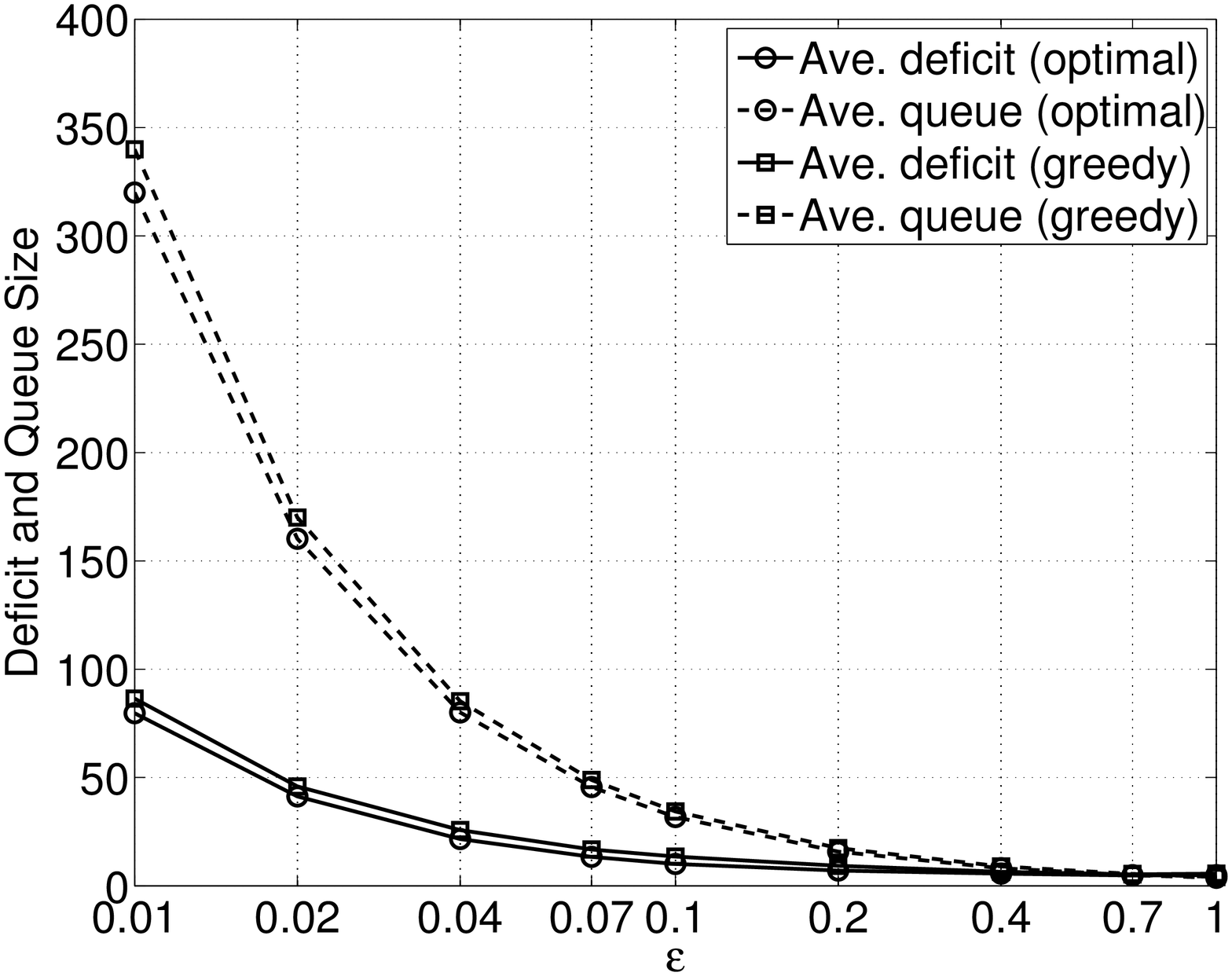}
	\caption{Deficit size and queue length when $w_l=3$}
	\label{def_queue_ten_w3}
\end{figure}
\begin{figure}[t]
	\centering
	\includegraphics[angle=270, width=3.3in]{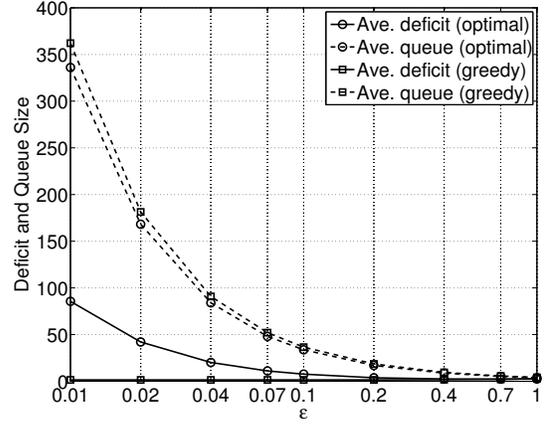}
	\caption{Deficit size and queue length when $w_l=6$}
	\label{def_queue_ten_w6}
\end{figure}

We see that as $w_l$ increases, the deficit counters become small. The upper bound in Corollary \ref{def_queue_bound_kc} only suggests that the sum of the deficit counters and queues is $O(1 / \epsilon)$. Thus, it is interesting to note that by changing $w_l$, one can nearly eliminate the backlog in deficit for inelastic traffic while maintaining the same order of queue sizes. The reason for this can be understood by examining the scheduling algorithm (\ref{online_opt_sch_kc}). Note that the algorithm gives priority to elastic traffic if queues are larger than counters. When $w_l$ is small compared to $\epsilon$, the effect of $w_l$ is negligible in the scheduling algorithm. On the other hand, when $w_l$ is $O(1)$, $w_l / \epsilon$ is $O(1 / \epsilon)$ which is comparable to the queue lengths and hence, the deficit does not have to be large to provide service to inelastic traffic under algorithm (\ref{online_opt_sch_kc}).

It must be noted that small deficit counters mean that there is a small backlog in providing acceptable service to inelastic arrivals. For the case of real-time traffic this is a desirable property, since we do not want to have large variations in the service provided that could affect the perceived quality. Thus, even if fair allocation of bandwidth beyond the minimum is not required for inelastic flows, choosing $w_l$ an order of magnitude larger than $\epsilon$ is desirable to maintain small deficits.

As can be noted, the greedy scheduler seems to give lower deficit values than the optimal scheduler for larger values of $w_l$. We believe that the reason is that weights given to inelastic flows increase with increasing $w_l$ and therefore, the greedy scheduler picks them first. However, our optimality goal is given by (\ref{offline_opt_kc}) which is determined by the rates received by the various flows. The rates achieved by the two schedulers are quite close in the simulations, as seen in Figs. \ref{rate_ten_w0}, \ref{rate_ten_w3}, and \ref{rate_ten_w6}, while keeping the dropping probabilities below the requirement, as shown in Figs. \ref{drop_prob_ten_w0}, \ref{drop_prob_ten_w3}, and \ref{drop_prob_ten_w6}.

\begin{figure}[t]
	\centering
	\includegraphics[angle=270, width=3.3in]{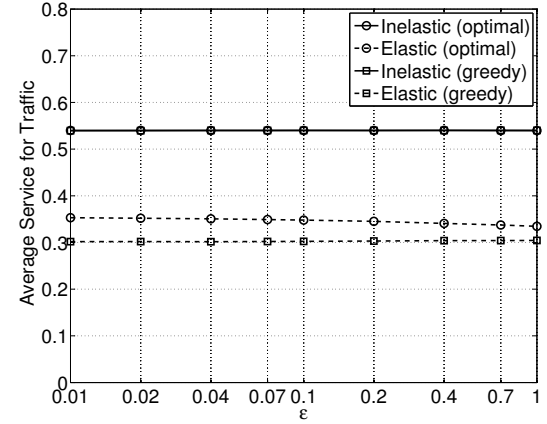}
	\caption{Average service when $w_l=0$}
	\label{rate_ten_w0}
\end{figure}
\begin{figure}[t]
	\centering
	\includegraphics[angle=270, width=3.3in]{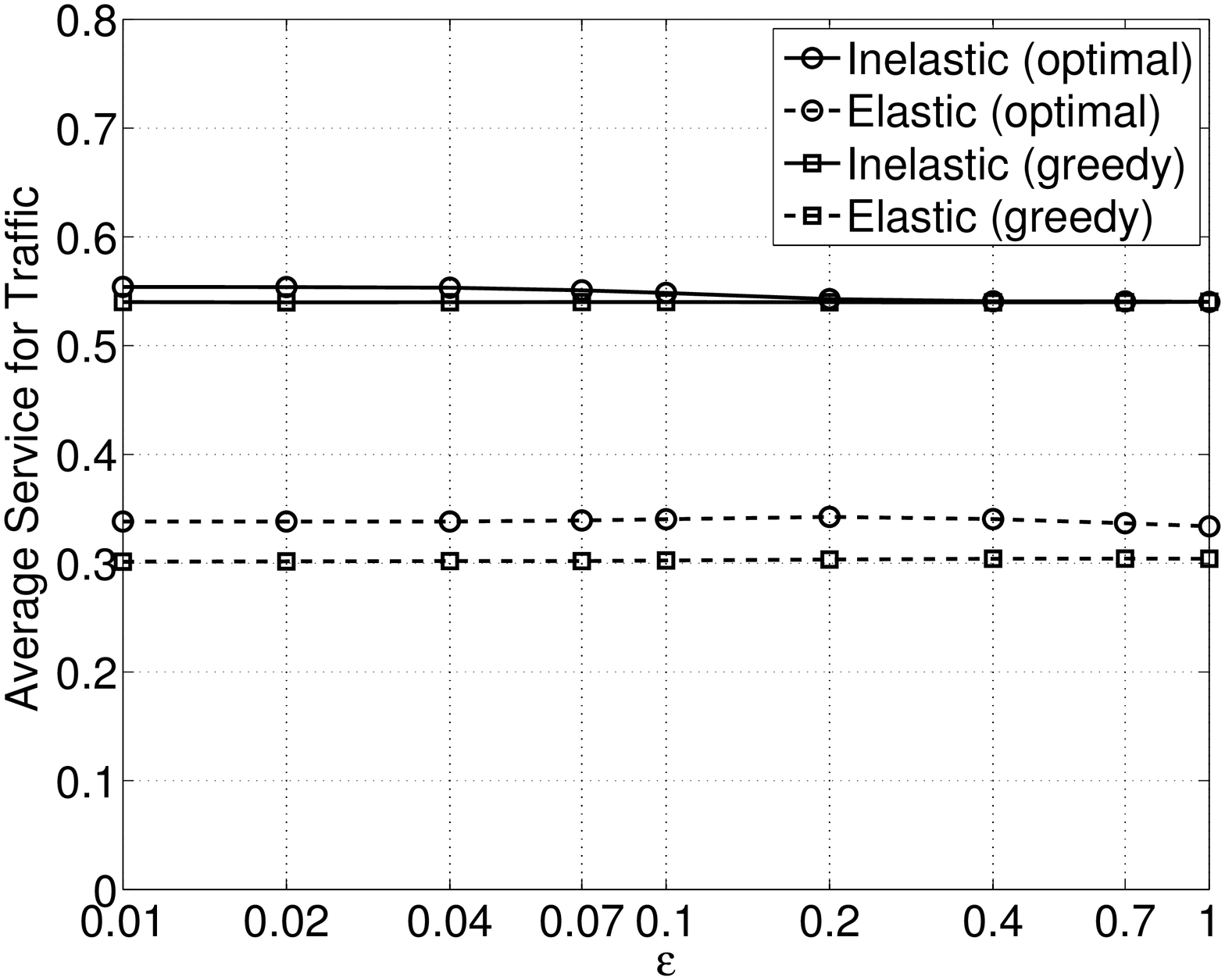}
	\caption{Average service when $w_l=3$}
	\label{rate_ten_w3}
\end{figure}
\begin{figure}[t]
	\centering
	\includegraphics[angle=270, width=3.3in]{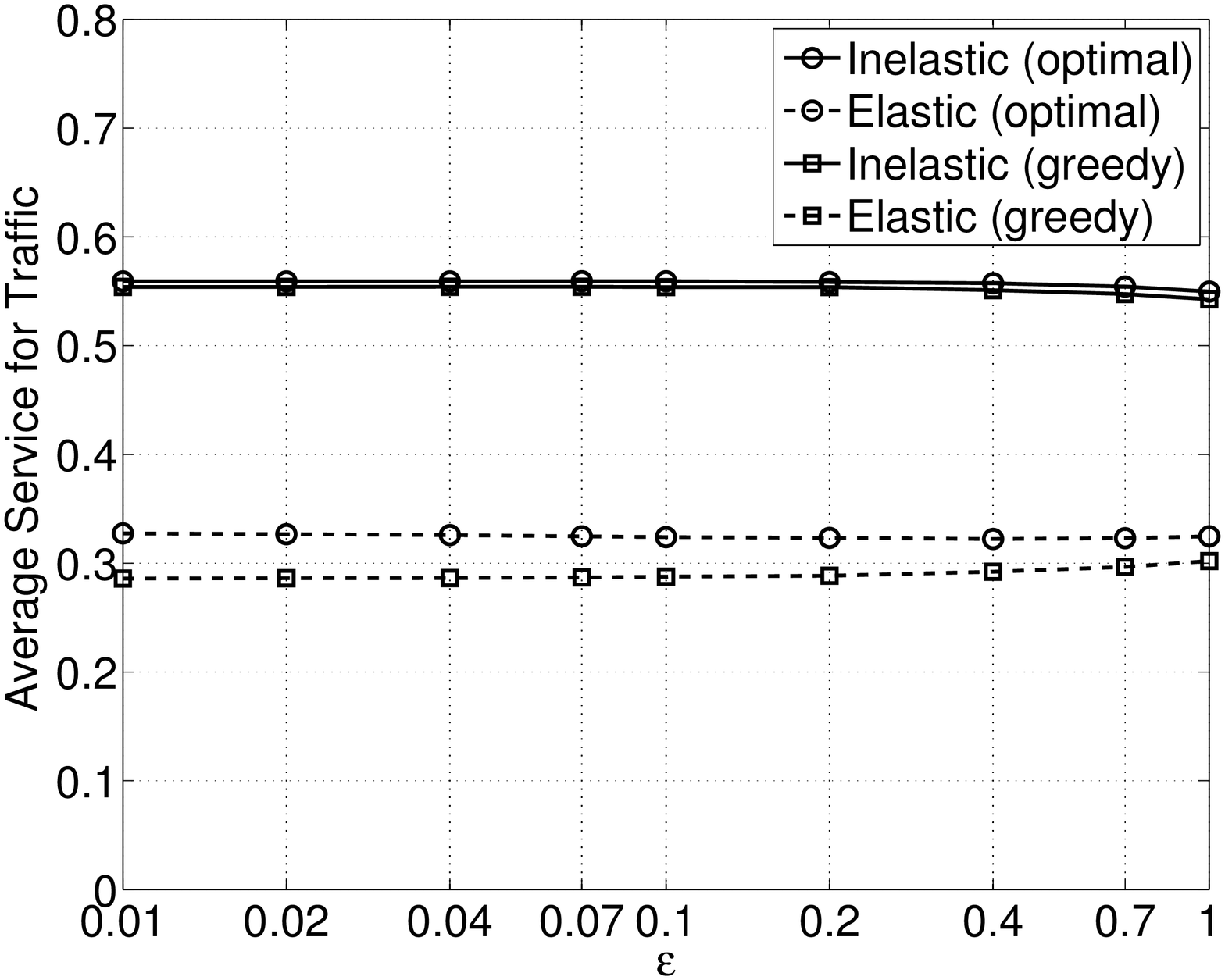}
	\caption{Average service when $w_l=6$}
	\label{rate_ten_w6}
\end{figure}

\begin{figure}[t]
	\centering
	\includegraphics[angle=270, width=3.3in]{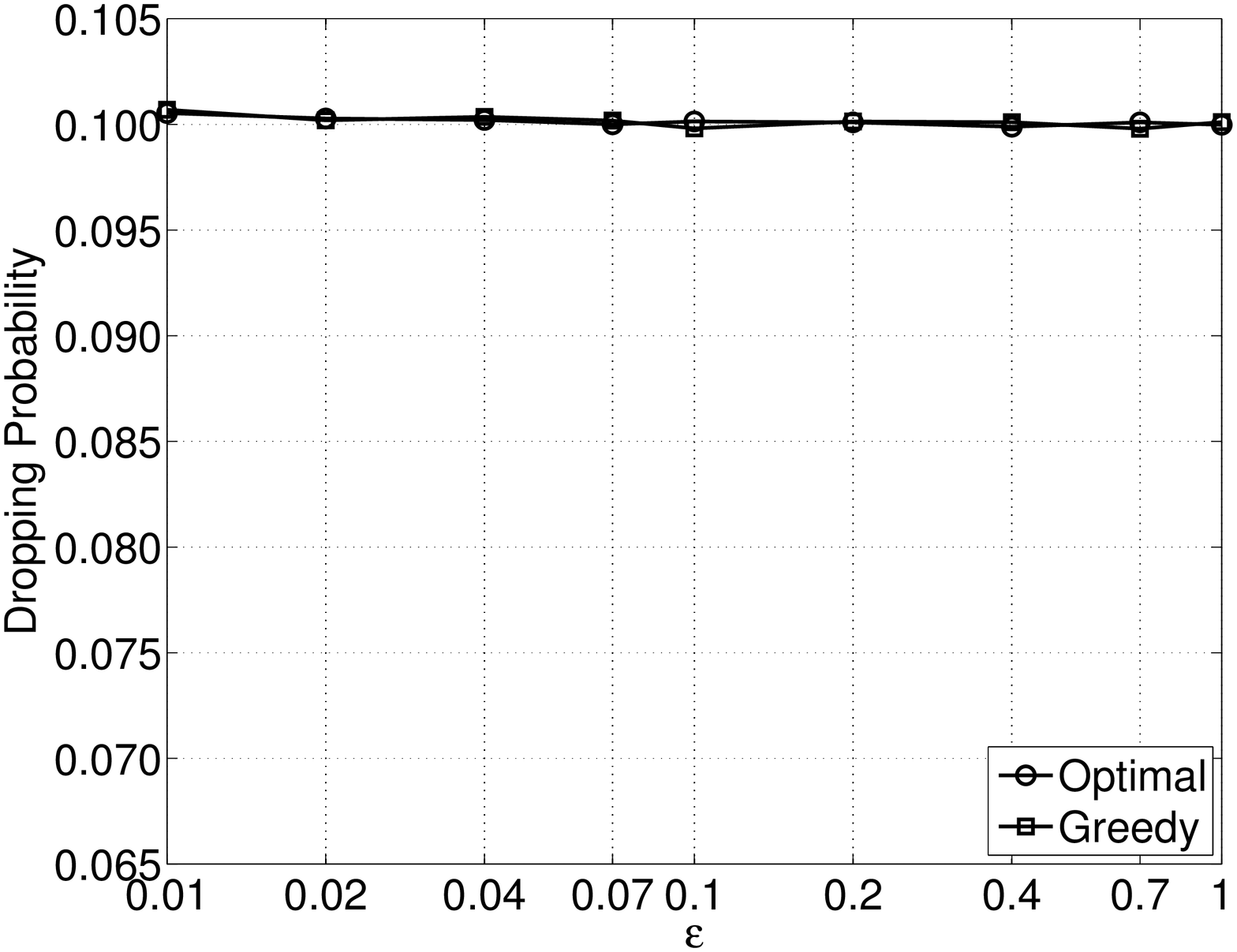}
	\caption{Dropping probability when $w_l=0$}
	\label{drop_prob_ten_w0}
\end{figure}
\begin{figure}[t]
	\centering
	\includegraphics[angle=270, width=3.3in]{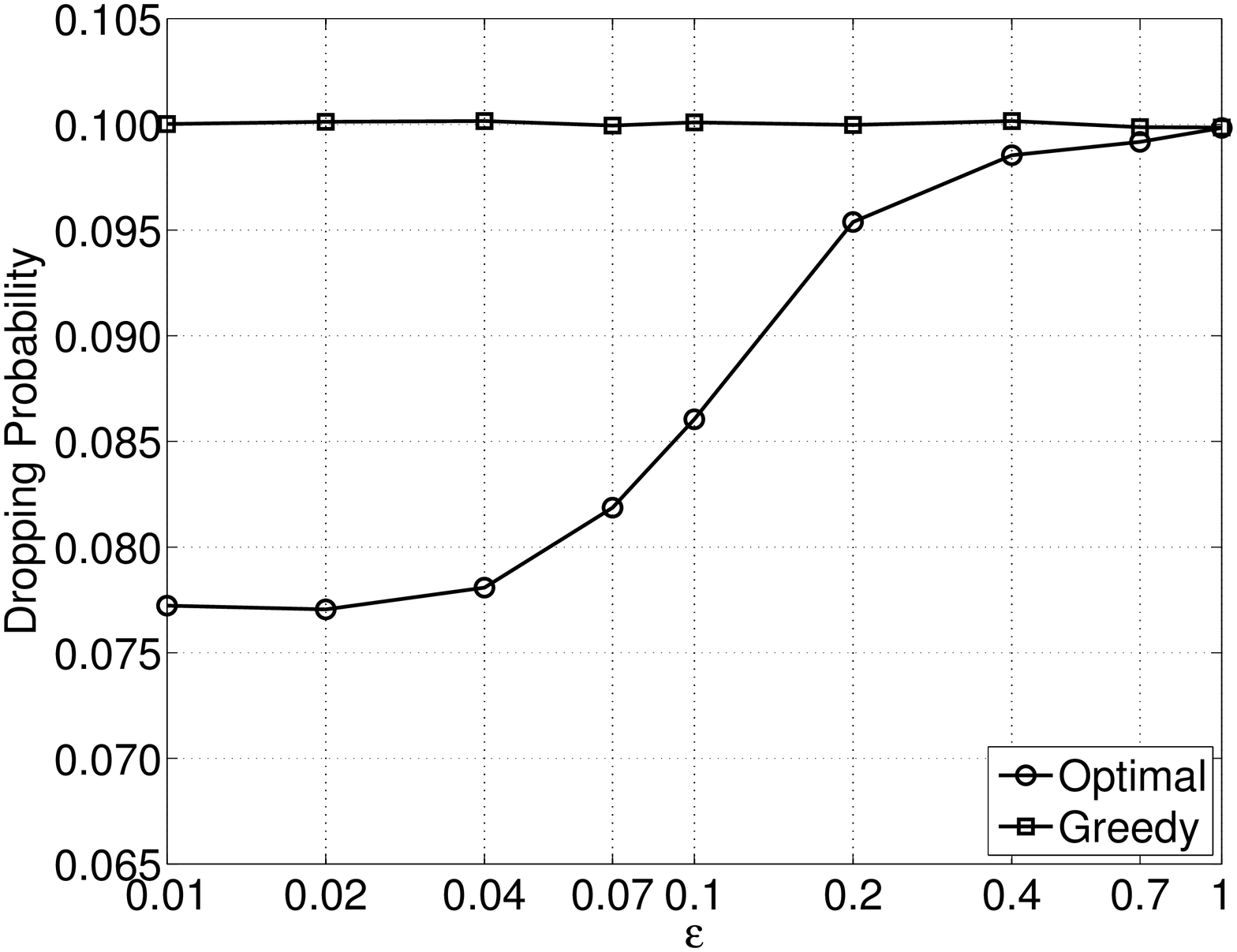}
	\caption{Dropping probability when $w_l=3$}
	\label{drop_prob_ten_w3}
\end{figure}
\begin{figure}[t]
	\centering
	\includegraphics[angle=270, width=3.3in]{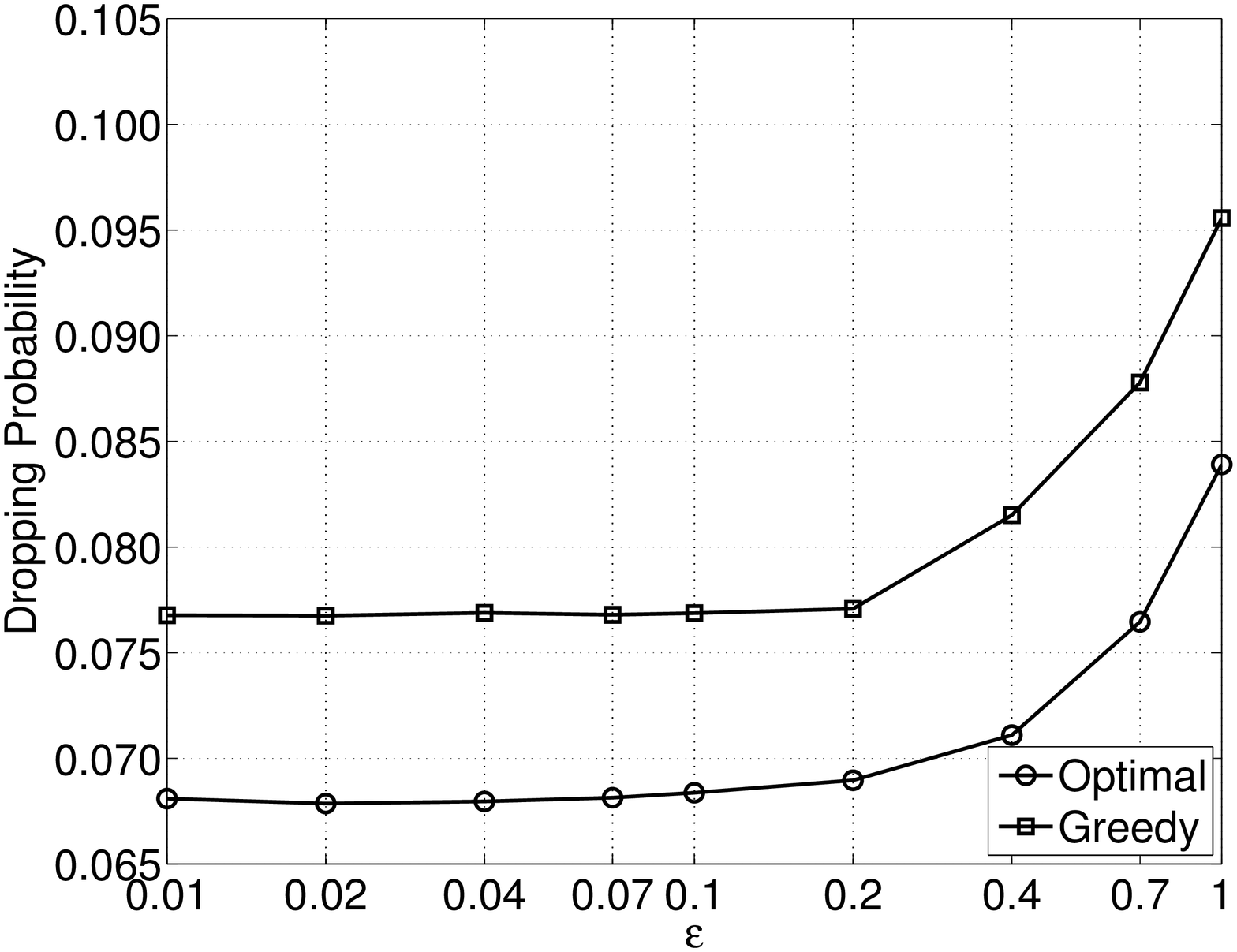}
	\caption{Dropping probability when $w_l=6$}
	\label{drop_prob_ten_w6}
\end{figure}

%
%
\section{Conclusions}

In this paper we have presented an optimization framework for the problem of congestion control and scheduling of elastic and inelastic traffic in ad hoc wireless networks. The model was developed for general interference graphs, general arrivals and time-varying channels. Using a dual function approach we presented a decomposition of the problem into an online algorithm that is able to make optimal decisions while keeping the network stable and fulfilling the inelastic flow's QoS constraints. A key result is that, through the use of deficit counters, one can treat the scheduling problem for elastic and inelastic flows in a common framework. It is also interesting to note that the deficit counters introduced in \cite{Hou09a, Hou09b, Hou09c} have the interpretation of Lagrange multipliers. Simulations corroborate our results and show the dependency of the performance of the algorithm on the auxiliary parameter $\epsilon$ and its role into assigning resources to both elastic and inelastic traffic. We note that, in the unknown channel case, we do not consider channel feedback at the end of each time slot as in \cite{Hou09a, Hou09b, Hou09c}. This will be addressed in future work.

%
%
\appendices
\section{Proof of Lemma \ref{expected_drift_kc}}

To prove Lemma \ref{expected_drift_kc}, we start by first proving two auxiliary lemmas and then stating a fact.

\begin{lemma}
\label{bound_d_kc}
Given that at frame $k$ we have the event $d(k)=d$ and $q(k)=q$, then
\begin{align*}
	E & \left[ \frac{1}{2} \sum_{l \in \mathcal{L}} \{ [ d_l + \tilde{a}_{il}(k) - I_{il}^*(a_i(k),c(k),d,q) ]^+ \}^2 \right] - \sum_{l \in \mathcal{L}} \frac{d_l^2}{2} \\
 	\leq & B_6 + \sum_{l \in \mathcal{L}} d_l \lambda_l (1-p_l) \notag \\
 	& - E \left[ \sum_{l \in \mathcal{L}} \left( \frac{1}{\epsilon}w_l + d_l \right) I_{il}^*(a_i(k),c(k),d,q) \right. \notag \\
 	& \left. - \sum_{l \in \mathcal{L}} \frac{1}{\epsilon}w_l I_{il}^*(a_i(k),c(k),d,q) \right] \notag \\
\end{align*}
for some non-negative constant $B_6$, and where $I_{il}^*(a_i(k),c(k),d,q)$ is given by the solution to (\ref{online_opt_sch_kc}).
$\hfill \diamond$
\end{lemma}

\begin{proof}
\begin{align}
 	E & \left[ \frac{1}{2} \sum_{l \in \mathcal{L}} \{ [ d_l + \tilde{a}_{il}(k) - I_{il}^*(a_i(k),c(k),d,q) ]^+ \}^2 \right] - \sum_{l \in \mathcal{L}} \frac{d_l^2}{2} \notag \\
	\leq & E \left[ \frac{1}{2} \sum_{l \in \mathcal{L}} [ d_l + \tilde{a}_{il}(k) - I_{il}^*(a_i(k),c(k),d,q) ]^2 \right] - \sum_{l \in \mathcal{L}} \frac{d_l^2}{2} \notag \\
 	= & E \left[ \sum_{l \in \mathcal{L}} d_l [ \tilde{a}_{il}(k) - I_{il}^*(a_i(k),c(k),d,q) ] \right. \notag \\
 	& \left. + \frac{1}{2} \sum_{l \in \mathcal{L}} [ \tilde{a}_{il}(k) - I_{il}^*(a_i(k),c(k),d,q) ]^2 \right] \notag \\
 	\leq & E \left[ \sum_{l \in \mathcal{L}} d_l \tilde{a}_{il}(k) - d_l I_{il}^*(a_i(k),c(k),d,q) \right. \notag \\
 	& \left. + \frac{1}{2} \sum_{l \in \mathcal{L}} \tilde{a}_{il}^2(k) + a_{il}^2(k) \right] \label{def_I_il} \\
 	\leq & B_6 + \sum_{l \in \mathcal{L}} d_l \lambda_l (1-p_l) \notag \\
 	& - E \left[ \sum_{l \in \mathcal{L}} \left( \frac{1}{\epsilon}w_l + d_l \right) I_{il}^*(a_i(k),c(k),d,q) \right. \notag \\
 	& \left. - \sum_{l \in \mathcal{L}} \frac{1}{\epsilon}w_l I_{il}^*(a_i(k),c(k),d,q) \right] \notag
\end{align}
where (\ref{def_I_il}) follows from the definition of $I_{il}^*(a_i(k),c(k),d,q)$ and
\begin{equation*}
	B_6 = \frac{1}{2} \sum_{l \in \mathcal{L}} ( \lambda_l^2 + \sigma_{il}^2 ) [ 1+(1-p_l)^2 ] + \lambda_l p_l (1-p_l).
\end{equation*}
\end{proof}

\begin{lemma}
\label{bound_q_kc}
Given that at frame $k$ we have the event $d(k)=d$ and $q(k)=q$, then
\begin{align*}
	E & \left[ \frac{1}{2} \sum_{l \in \mathcal{L}} \{ [ q_l + \tilde{a}_{el}(k) - I_{el}^*(a_i(k),c(k),d,q) ]^+ \}^2 \right] - \sum_{l \in \mathcal{L}} \frac{q_l^2}{2} \\
 	\leq & B_7 - \frac{1}{\epsilon} \sum_{l \in \mathcal{L}} [ U_l(x_l^*) - U_l(\tilde{x}_l^*(k)) ] \notag \\
 	& - \sum_{l \in \mathcal{L}} q_l \left\{ E \left[ I_{el}^*(a_i(k),c(k),d,q) \right] - x_l^* \right\} \notag
\end{align*}
for some constant $B_7>0$, where $x^*$ and $\tilde{x}^*(k)$ are the solutions to (\ref{offline_opt_kc}) and (\ref{online_opt_cc_kc}) respectively.
$\hfill \diamond$
\end{lemma}

\begin{proof}
\begin{align}
 	E & \left[ \frac{1}{2} \sum_{l \in \mathcal{L}} \{ [ q_l + \tilde{a}_{el}(k) - I_{el}^*(a_i(k),c(k),d,q) ]^+ \}^2 \right] - \sum_{l \in \mathcal{L}} \frac{q_l^2}{2} \notag \\
	\leq & E \left[ \frac{1}{2} \sum_{l \in \mathcal{L}} [ q_l + \tilde{a}_{el}(k) - I_{el}^*(a_i(k),c(k),d,q) ]^2 \right] - \sum_{l \in \mathcal{L}} \frac{q_l^2}{2} \notag \\
 	= & E \left[ \sum_{l \in \mathcal{L}} q_l [ \tilde{a}_{el}(k) - I_{el}^*(a_i(k),c(k),d,q) ] \right. \notag \\
 	& \left. + \frac{1}{2} \sum_{l \in \mathcal{L}} [ \tilde{a}_{el}(k) - I_{el}^*(a_i(k),c(k),d,q) ]^2 \right] \notag \\
 	\leq & E \left[ \sum_{l \in \mathcal{L}} q_l \tilde{a}_{el}(k) - q_l I_{el}^*(a_i(k),c(k),d,q) \right. \notag \\
 	& \left. + \frac{1}{2} \sum_{l \in \mathcal{L}} ( \tilde{a}_{el}^2(k) + c_l^2 T^2 ) \right] \label{def_I_el} \\
 	\leq & B_7 + \sum_{l \in \mathcal{L}} - [ \frac{1}{\epsilon}U_l(\tilde{x}_l^*(k)) - q_l \tilde{x}_l^*(k) ] + \sum_{l \in \mathcal{L}} \frac{1}{\epsilon}U_l(\tilde{x}_l^*(k)) \notag \\
 	& - \sum_{l \in \mathcal{L}} q_l E \left[ I_{el}^*(a_i(k),c(k),d,q) \right] \notag \\
 	\leq & B_7 + \sum_{l \in \mathcal{L}} - [ \frac{1}{\epsilon}U_l(x_l^*) - q_l x_l^* ] + \sum_{l \in \mathcal{L}} \frac{1}{\epsilon}U_l(\tilde{x}_l^*(k)) \label{def_opt_point} \\
 	& - \sum_{l \in \mathcal{L}} q_l E \left[ I_{el}^*(a_i(k),c(k),d,q) \right] \notag \\
 	= & B_7 - \frac{1}{\epsilon} \sum_{l \in \mathcal{L}} [ U_l(x_l^*) - U_l(\tilde{x}_l^*(k)) ] \notag \\
 	& - \sum_{l \in \mathcal{L}} q_l \left\{ E \left[ I_{el}^*(a_i(k),c(k),d,q) \right] - x_l^* \right\} \notag
\end{align}
where (\ref{def_I_el}) follows from the definition of $I_{el}^*(a_i(k),c(k),d,q)$,
\begin{equation*}
	B_7 = \frac{1}{2} \sum_{l \in \mathcal{L}} X_{max}^2 + \sigma_{e}^2 + ( \bar{c}_l^2 + \sigma_{cl}^2 ) T^2,
\end{equation*}
and (\ref{def_opt_point}) follows from the fact that $\tilde{x}^*(k)$ is the optimal point of (\ref{online_opt_cc_kc}).
\end{proof}

\begin{fact}
\label{fact_online_opt_kc}
The optimization in (\ref{online_opt_sch_kc}) can be performed over $\mathcal{S}(a_i(k),c(k))_\mathcal{CH}$, the convex hull of $\mathcal{S}(a_i(k),c(k))$; that is,
\begin{align*}
	& \max\limits_{s \in \mathcal{S}(a_i(k),c(k))} \sum\limits_{l \in \mathcal{L}} [\frac{1}{\epsilon} w_l+d_l(k)] \sum_{t=1}^{T} s_{il,t} + q_l(k) \sum_{t=1}^{T} s_{el,t} = \\
	& \max\limits_{s \in \mathcal{S}(a_i(k),c(k))_\mathcal{CH}} \sum\limits_{l \in \mathcal{L}} [\frac{1}{\epsilon} w_l+d_l(k)] \sum_{t=1}^{T} s_{il,t} + q_l(k) \sum_{t=1}^{T} s_{el,t}.
\end{align*}
The reason for this comes from the fact that the objective function is linear and therefore there must be an optimal point $s^*(a_i(k),c(k),d(k),q(k)) ) \in \mathcal{S}(a_i(k),c(k))$.
$\hfill \diamond$
\end{fact}

\begin{proof}[Proof of Lemma \ref{expected_drift_kc}]
For the purpose of this proof, we define the capacity region for fixed arrival and channel states $a_i$ and $c$ as follows:
\begin{equation*}
	\mathcal{C}(a_i,c) \stackrel{def}{=}
	\left\{
	\begin{array}{l}
		( \bar{\mu}_{il}, \bar{\mu}_{el} )_{l \in \mathcal{L}}: \mbox{there exists } \bar{s} \in \mathcal{S}(a_i,c)_\mathcal{CH} \mbox{,} \\
		\bar{\mu}_{il}\leq\sum_{t=1}^{T} \bar{s}_{il,t} \mbox{ and } \bar{\mu}_{el}\leq\sum_{t=1}^{T} \bar{s}_{el,t}
	\end{array}
	\right\}.
\end{equation*}
Then, the overall capacity of the network is defined as
$\mathcal{C} \stackrel{def}{=}$
\begin{equation*}
	\left\{
	\begin{array}{l}
		( \mu_{il}, \mu_{el} )_{l \in \mathcal{L}} : \mbox{there exists } ( \bar{\mu}_{il}(a_i,c), \bar{\mu}_{el}(a_i,c) )_{l \in \mathcal{L}} \in \\
		\mathcal{C}(a_i,c) \mbox{ for all } a_i, c \mbox{ and } \mu_{il} = E[ \bar{\mu}_{il}(a_i,c) ] \mbox{, } \\
		\mu_{el} = E[ \bar{\mu}_{el}(a_i,c) ] \mbox{ for all } l \in \mathcal{L} \\
	\end{array}
	\right\}.
\end{equation*}

From Lemmas \ref{bound_d_kc} and \ref{bound_q_kc} we have:
\begin{align}
	E & \left[ V(d(k+1),q(k+1)) | d(k)=d, q(k)=q \right] - V(d,q) \notag \\
 	\leq & B_1 + \sum_{l \in \mathcal{L}} d_l \lambda_l (1-p_l) + \sum_{l \in \mathcal{L}} q_l x_l^* \notag \\
 	& - \frac{1}{\epsilon} \sum_{l \in \mathcal{L}} [ U_l(x_l^*) - U_l(\tilde{x}_l^*(k)) ] \notag \\
 	& - E \left[ \sum_{l \in \mathcal{L}} \left( \frac{1}{\epsilon}w_l + d_l \right) I_{il}^*(a_i(k),c(k),d,q) \right. \notag \\
 	& \left. + \sum_{l \in \mathcal{L}} q_l I_{el}^*(a_i(k),c(k),d,q) \right] \notag \\
 	& + \sum_{l \in \mathcal{L}} \frac{1}{\epsilon}w_l E \left[ I_{il}^*(a_i(k),c(k),d,q) \right] \notag \\
 	\leq & B_1 + \sum_{l \in \mathcal{L}} d_l \lambda_l (1-p_l) + \sum_{l \in \mathcal{L}} q_l x_l^* \label{fact_consequence} \\
 	& - \frac{1}{\epsilon} \sum_{l \in \mathcal{L}} [ U_l(x_l^*) - U_l(\tilde{x}_l^*(k)) ] \notag \\
 	& - E \left[ \sum_{l \in \mathcal{L}} \left( \frac{1}{\epsilon}w_l + d_l \right) \bar{\mu}_{il}(a_i(k),c(k)) \right. \notag \\
 	& \left. + \sum_{l \in \mathcal{L}} q_l \bar{\mu}_{el}(a_i(k),c(k)) \right] \notag \\
 	& + \sum_{l \in \mathcal{L}} \frac{1}{\epsilon}w_l E \left[ I_{il}^*(a_i(k),c(k),d,q) \right] \notag \\
 	= & B_1 - \sum_{l \in \mathcal{L}} d_l [ \mu_{il} - \lambda_l (1-p_l) ] - \sum_{l \in \mathcal{L}} q_l (\mu_{el} - x_l^*) \label{mean_mu_kc} \\
 	& - \frac{1}{\epsilon} \sum_{l \in \mathcal{L}} [ U_l(x_l^*) - U_l(\tilde{x}_l^*(k)) ] \notag \\
 	& - \frac{1}{\epsilon} \sum_{l \in \mathcal{L}} w_l \mu_{il} - w_l E \left[ I_{il}^*(a_i(k),c(k),d,q) \right] \notag
\end{align}
where $B_1=B_6+B_7$, (\ref{fact_consequence}) follows for any $( \bar{\mu}_{il}(a_i(k),c(k)), \bar{\mu}_{el}(a_i(k),c(k)) )_{l \in \mathcal{L}} \in \mathcal{C}(a_i(k),c(k))$ as was explained in Fact \ref{fact_online_opt_kc}, and (\ref{mean_mu_kc}) holds for any $( \mu_{il}, \mu_{el} )_{l \in \mathcal{L}} \in \mathcal{C}$. It should be clear that $( \mu^*_{il}, \mu^*_{el} )_{l \in \mathcal{L}} \in \mathcal{C}$, where $( \mu^*_{il}, \mu^*_{el} )_{l \in \mathcal{L}}$ is the solution to (\ref{offline_opt_kc}). Thus we have the following:
\begin{align*}
	E & \left[ V(d(k+1),q(k+1)) | d(k)=d, q(k)=q \right] - V(d,q) \\
 	\leq & B_1 - B_2 \sum_{l \in \mathcal{L}} d_l - B_3 \sum_{l \in \mathcal{L}} q_l - \frac{1}{\epsilon} \sum_{l \in \mathcal{L}} [ U_l(x_l^*) - U_l(\tilde{x}_l^*(k)) ] \\
 	& - \frac{1}{\epsilon} \sum_{l \in \mathcal{L}} w_l \mu^*_{il} - w_l E \left[ I_{il}^*(a_i(k),c(k),d,q) \right]
\end{align*}
where
\begin{equation*}
	B_2 = \min_{l \in \mathcal{L}} \left\{ \mu^*_{il} - \lambda_l (1-p_l) \right\}
\end{equation*}
and
\begin{equation*}
	B_3 = \min_{l \in \mathcal{L}} \left\{ \mu^*_{el} - x_l^* \right\}.
\end{equation*}
\end{proof}

\section{Proof of Theorem \ref{optimality_online_kc}}

From Lemma \ref{expected_drift_kc} we know that
\begin{align*}
	\frac{1}{\epsilon} & \sum_{l \in \mathcal{L}} U_l(x_l^*) - U_l(\tilde{x}_l^*(k)) \\
	+ & \frac{1}{\epsilon} \sum_{l \in \mathcal{L}} w_l \mu^*_{il} - w_l E \left[ I_{il}^*(a_i(k),c(k),d,q) \right] \\
 	\leq & B_1 - B_2 \sum_{l \in \mathcal{L}} d_l - B_3 \sum_{l \in \mathcal{L}} q_l + V(d,q) \\
 	& - E \left[ V(d(k+1),q(k+1)) | d(k)=d, q(k)=q \right] \\
 	\leq & B_1 + V(d,q) \\
 	& - E \left[ V(d(k+1),q(k+1)) | d(k)=d, q(k)=q \right]
\end{align*}
since $B_2 \sum_{l \in \mathcal{L}} d_l + B_3 \sum_{l \in \mathcal{L}} q_l \geq 0$. Taking expectations:
\begin{align*}
	\frac{1}{\epsilon} & E \left[ \sum_{l \in \mathcal{L}} U_l(x_l^*) - U_l(\tilde{x}_l^*(k)) \right. \\
	& \left. + \sum_{l \in \mathcal{L}} w_l \mu^*_{il} - w_l I_{il}^*(a_i(k),c(k),d(k),q(k)) \right] \\
	\leq & B_1 - E\left[ V(d(k+1),q(k+1)) \right] + E\left[ V(d(k),q(k)) \right].
\end{align*}

Adding the terms for $k=\{ 1, \ldots, K \}$ and dividing by $K$ we get:
\begin{align}
	\frac{1}{\epsilon} & E \left[ \sum_{l \in \mathcal{L}} U_l(x_l^*) + w_l \mu^*_{il} \right. \notag \\
	& \left. - \sum_{l \in \mathcal{L}} \frac{1}{K} \sum_{k=1}^K U_l(\tilde{x}_l^*(k)) + w_l I_{il}^*(a_i(k),c(k),d(k),q(k)) \right] \notag \\
	& \leq B_1 - \frac{E\left[ V(d(K+1),q(K+1)) \right]}{K} + \frac{E\left[ V(d(1),q(1)) \right]}{K} \notag \\
	& \leq B_1 + \frac{E\left[ V(d(1),q(1)) \right]}{K} \label{positive_mean_kc}
\end{align}
where (\ref{positive_mean_kc}) follows from the fact that the Lyapunov function V is non-negative.

Assuming $E\left[ V(d(1),q(1)) \right] < \infty$ we get the following limit expression:
\begin{align*}
	\limsup_{K \rightarrow \infty} & E \left[ \sum_{l \in \mathcal{L}} U_l(x_l^*) + w_l \mu^*_{il} - \sum_{l \in \mathcal{L}} \frac{1}{K} \sum_{k=1}^K U_l(\tilde{x}_l^*(k)) \right. \\
	& \left. - \sum_{l \in \mathcal{L}} \frac{1}{K} \sum_{k=1}^K w_l I_{il}^*(a_i(k),c(k),d(k),q(k)) \right] \leq B \epsilon
\end{align*}
where $B=B_1$.\hfill $\blacksquare$

%
%
%
%
%


%


%
%
\end{document}